\documentclass[a4paper, runningheads]{llncs}
\usepackage{amssymb,amsfonts,amsmath}
\usepackage{booktabs,tabularx}
\usepackage{url,xspace,soul,wrapfig}
\usepackage{graphicx,paralist,enumerate}
\usepackage{hyperref}
\usepackage{subfigure}
\usepackage[textsize=footnotesize]{todonotes}
\usepackage[basic]{complexity}
\usepackage[capitalise]{cleveref}

\hypersetup{
hidelinks,
colorlinks=true,
citecolor=[rgb]{0.121 0.47 0.705},
linkcolor=[rgb]{0.121 0.47 0.705},
urlcolor=[rgb]{0.121 0.47 0.705}
}

\newcommand{\bigoh}{\mathcal{O}}

\newcommand{\RRR}{\mathcal{R}}

\newtheorem{observation}{Observation}

\newcommand{\FPT}{\mbox{{\sf FPT}}}

\newcommand{\old}[1]{{}}

\DeclareMathOperator{\bt}{bt}

\DeclareMathOperator{\fobt}{fo-bt}

\title{Parameterized Algorithms\\ for Book Embedding Problems\thanks{Research of Fabrizio Montecchiani supported in part by MIUR under Grant 20174LF3T8 AHeAD: efficient Algorithms for HArnessing networked Data. Robert Ganian acknowledges support by the FWF (Project P31336, ``NFPC'') and is also affiliated with FI MUNI, Brno, Czech Republic.
Research of Sujoy Bhore and Martin N\"ollenburg is supported by the Austrian Science Fund (FWF) grant P 31119.}}
\author{Sujoy Bhore\inst{1} \and Robert Ganian\inst{1} \and Fabrizio Montecchiani\inst{2} \and\\ Martin N\"ollenburg\inst{1} }

\institute{Algorithms and Complexity Group, TU Wien, Vienna, Austria\\\email{\{sujoy,rganian,noellenburg\}@ac.tuwien.ac.at}
\and
Engineering Department, University of Perugia, Perugia, Italy\\\email{fabrizio.montecchiani@unipg.it}
}

\newcommand*{\claimproofs}{\emph{Proof of the Claim.}\quad}

\begin{document}
\maketitle

\begin{abstract}
A \emph{$k$-page book embedding} of a graph $G$ draws the vertices of $G$ on a line and the edges on $k$ half-planes (called \emph{pages}) bounded by this line, such that no two edges on the same page cross. We study the problem of determining whether $G$ admits a $k$-page book embedding both when the linear order of the vertices is fixed, called \textsc{Fixed-Order Book Thickness}, or not fixed, called \textsc{Book Thickness}. Both problems are known to be \NP-complete in general. We show that \textsc{Fixed-Order Book Thickness} and \textsc{Book Thickness} are fixed-parameter tractable parameterized by the vertex cover number of the graph and that \textsc{Fixed-Order Book Thickness} is fixed-parameter tractable parameterized by the pathwidth of the vertex order.  
\end{abstract}

\section{Introduction}
A \emph{$k$-page book embedding} of a graph $G$ is a drawing that maps the vertices of $G$ to distinct points on a line, called  \emph{spine}, and each edge to a simple curve drawn inside one of $k$ half-planes bounded by the spine, called \emph{pages}, such that no two edges on the same page cross~\cite{kainen74,ollmann73}; see \cref{fig:intro} for an illustration. 
This kind of layout can be alternatively defined in combinatorial terms as follows. 
A $k$-page book embedding of $G$ is a linear order $\prec$ of its vertices and a coloring of its edges which guarantee that no two edges $uv$, $wx$ of the same color have their vertices ordered as $u \prec w \prec v \prec x$. 
The minimum $k$ such that $G$ admits a $k$-page book embedding is the \emph{book thickness} of $G$, denoted by $\bt(G)$, also known as the \emph{stack number} of $G$. 
Book embeddings have been extensively studied in the literature, among others due to their applications in bioinformatics, VLSI, and parallel computing (see, e.g.,~\cite{doi:10.1137/0608002,Haslinger1999} and refer also to~\cite{DBLP:journals/dmtcs/DujmovicW04} for a survey). 
A famous result by Yannakakis~\cite{DBLP:journals/jcss/Yannakakis89} states that every planar graph has book thickness at most four. 
Several other bounds are known for special graph families, for instance planar graphs with vertex degree at most four have book thickness two~\cite{DBLP:journals/algorithmica/BekosGR16}, while graphs of treewidth $w>2$ have book thickness $w+1$~\cite{DBLP:journals/dcg/DujmovicW07,DBLP:journals/dam/GanleyH01}. 

\begin{figure}[t]
\centering
\subfigure[~]{\includegraphics[width=0.3\linewidth,page=1]{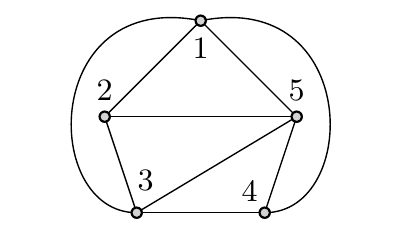}}
\hfil
\subfigure[~]{\includegraphics[width=0.3\linewidth,page=2]{figs/intro}}
\hfil
\subfigure[~]{\includegraphics[width=0.3\linewidth,page=3]{figs/intro}}
\vspace{-0.3cm}
\caption{(a) A planar graph $G$ with book thickness two. (b) A $2$-page book embedding of $G$. (c) A linear order of $G$ such that its fixed-order book thickness is three (and the corresponding $3$-page book embedding). \vspace{-0.5cm}}
\label{fig:intro}
\end{figure}

Given a graph $G$ and a positive integer $k$, the problem of determining whether $\bt(G) \le k$, called \textsc{Book Thickness}, is known to be \NP-complete. Namely, Bernhart and Kainen~\cite{DBLP:journals/jct/BernhartK79} proved that $\bt(G) \le 2$ if and only if $G$ is subhamiltonian, i.e., $G$ is a subgraph of a planar Hamiltonian graph. Since deciding whether a graph is subhamiltonian is an \NP-complete problem, \textsc{Book Thickness} is also \NP-complete in general~\cite{doi:10.1137/0608002}. \textsc{Book Thickness} has been studied also when the linear order $\prec$ of the vertices is fixed, indeed, this is one of the original formulations of the problem, which arises in the context of sorting with parallel stacks~\cite{doi:10.1137/0608002}. We call this problem \textsc{Fixed-Order Book Thickness} and we denote by $\fobt(G,\prec)$ the \emph{fixed-order book thickness} of a graph $G$.
Obviously, we have $\fobt(G,\prec) \ge \bt(G)$, see \cref{fig:intro}. 
Deciding whether $\fobt(G,\prec) \le 2$ corresponds to testing the bipartiteness of a suitable conflict graph, and thus it can be solved in linear time. On the other hand, deciding if $\fobt(G,\prec) \le 4$  is equivalent to finding a 4-coloring of a circle graph and hence is an \NP-complete problem~\cite{DBLP:conf/stacs/Unger92}.

\smallskip\noindent\textbf{Our Results.} 
In this paper we study the parameterized complexity of \textsc{Book Thickness} and \textsc{Fixed-Order Book Thickness}. 
For both problems, when the answer is positive, we naturally also expect to be able to compute a corresponding $k$-page book embedding as a witness. While both problems are \NP-complete already for small fixed values of $k$ on general graphs, it is natural to ask which structural properties of the input (formalized in terms of structural parameters) allow us to solve these problems efficiently. 
Indeed, already Dujmovic and Wood~\cite{DBLP:journals/dmtcs/WoodD11} asked  whether \textsc{Book Thickness} can be solved in polynomial time when the input graph has bounded treewidth~\cite{DBLP:journals/jal/RobertsonS86}---a question which has turned out to be surprisingly resilient to existing algorithmic techniques and remains open to this day. Bannister and Eppstein~\cite{DBLP:journals/jgaa/BannisterE18} made partial progress towards answering Dujmovic and Wood's question by showing that \textsc{Book Thickness} is fixed-parameter tractable parameterized by the treewidth of $G$ when $k=2$.

We provide the first fixed-parameter algorithms for \textsc{Fixed-Order Book Thickness} and also the first such algorithm for \textsc{Book Thickness} that can be used when $k>2$. In particular, we provide fixed-parameter algorithms for:

\begin{enumerate}
\item \textsc{Fixed-Order Book Thickness} parameterized by the vertex cover number of the graph;
\item \textsc{Fixed-Order Book Thickness} parameterized by the pathwidth of the graph and the vertex order; and
\item \textsc{Book Thickness} parameterized by the vertex cover number of the graph.
\end{enumerate}

Results 1 and 2 are obtained by combining dynamic programming techniques with insights about the structure of an optimal book embedding. Result 3 then applies a kernelization technique to obtain an equivalent instance of bounded size (which can then be solved, e.g., by brute force). All three of our algorithms can also output a corresponding $k$-page book embedding as a witness (if it exists).

The remainder of this paper is organized as follows. \cref{preliminaries} contains preliminaries and basic definitions. Results~1 and~2 on \textsc{Fixed-Order Book Thickness} are presented in \cref{fixed-ordering}, while Result~3 on \textsc{Book Thickness} is described in \cref{without-ordering}. 
Conclusions and open problems are found in \cref{conclusions}. 
Statements with a proof in the appendix are marked by an asterisk (*).

\section{Preliminaries}\label{preliminaries}
We use standard terminology from graph theory~\cite{DBLP:books/daglib/0030488}.
For $r \in \mathbb{N}$, we write $[r]$ as shorthand for the set $\{1, \ldots, r\}$.
Parameterized complexity~\cite{DBLP:books/sp/CyganFKLMPPS15,DBLP:series/txcs/DowneyF13} focuses on the study of problem complexity not only with respect to the input size $n$ but also a parameter $k \in \mathbb{N}$.
The most desirable complexity class in this setting is 
\FPT\ (\emph{fixed-parameter tractable}), which contains all problems that can be solved by an algorithm
running in time $f(k)\cdot n^{\bigoh(1)}$, where $f$ is a computable
function. Algorithms running in this time are called \emph{fixed-parameter algorithms}.

A $k$-page book embedding of a graph $G=(V,E)$ will be denoted by a pair $\langle \prec, \sigma \rangle$, where $\prec$ is a linear order of $V$, and $\sigma \colon E \rightarrow [k]$ is a function that maps each edge of $E$ to one of $k$ pages $[k] = \{1, 2, \dots, k\}$. In a $k$-page book embedding $\langle \prec, \sigma \rangle$ it is required that for no pair of edges $uv, wx \in E$ with $\sigma(uv) = \sigma(wx)$ the vertices are ordered as $u \prec w \prec v \prec x$, i.e., each page is crossing-free. 

We consider two graph parameters for our algorithms. 
A \emph{vertex cover} $C$ of a graph $G=(V,E)$ is a subset  $C \subseteq V$ such that each edge in $E$ has at least one end-vertex in $C$. The \emph{vertex cover number} of $G$, denoted by $\tau(G)$, is the size of a minimum vertex cover of $G$.
The second parameter is \emph{pathwidth},
a classical graph parameter~\cite{DBLP:journals/jct/RobertsonS83} 
which admits several equivalent definitions. The definition that will be most useful here
is the one tied to linear orders~\cite{DBLP:journals/ipl/Kinnersley92}; see also~\cite{LodhaOS17,Mallach17} for recent works using this formulation.
Given an $n$-vertex graph $G=(V,E)$ with a linear order $\prec$ of $V$ such that $v_1\prec v_2 \prec \dots \prec v_n$, the \emph{pathwidth} of $(G,\prec)$ is the minimum number $\kappa$ such that for each vertex $v_i$ ($i\in [n]$), there are at most $\kappa$ vertices left of $v_i$ that are adjacent to $v_i$ or a vertex right of $v_i$.
Formally, for each $v_i$ we call the set $P_i=\{ v_j \mid j<i, \exists q \geq i$ such that $v_j v_q \in E\}$ the \emph{guard set} for $v_i$, and the pathwidth of $(G,\prec)$
is simply $\max_{i\in [n]} |P_i|$.
The elements of the guard sets are called the \emph{guards} (for $v_i$).
We remark that the pathwidth of $G$ is equal to the minimum pathwidth over all linear orders $\prec$.

\section{Algorithms for \textsc{Fixed-Order Book Thickness} }\label{fixed-ordering}
Recall that in \textsc{Fixed-Order Book Thickness} the input consists of a graph $G=(V,E)$, a linear order $\prec$ of $V$, and a positive integer $k$. We assume that $V=\{v_1, v_2, \dots, v_n\}$ is indexed such that $i < j \Leftrightarrow v_i \prec v_j$. 
The task is to decide if there is a page assignment $\sigma \colon E \rightarrow [k]$ such that $\langle \prec, \sigma \rangle$ is a $k$-page book embedding of $G$, i.e., whether $\fobt(G,\prec) \le k$. If the answer is `YES' we shall return a corresponding $k$-page book embedding as a witness. In fact, our algorithms will return a book embedding with the minimum number of pages.

\subsection{Parameterization by the Vertex Cover Number}\label{sec:vc_fixed}

As our first result, we will show that \textsc{Fixed-Order Book Thickness} is fixed-parameter tractable when parameterized by the \emph{vertex cover number}. 
We note that the vertex cover number is a graph parameter which, while restricting the structure of the graph in a fairly strong way, has been used to obtain fixed-parameter algorithms for numerous difficult problems~\cite{DBLP:journals/jgaa/BannisterCE18,DBLP:conf/isaac/FellowsLMRS08,DBLP:journals/dmtcs/Ganian15}.

Let $C$ be a minimum vertex cover of size $\tau = \tau(G)$;
we remark that such a vertex cover $C$ can be computed in time $\bigoh(2^\tau+\tau\cdot n)$~\cite{DBLP:journals/tcs/ChenKX10}. 
Moreover, let $U=V\setminus C$.
Our first observation shows that the problem becomes trivial if $\tau \le k$.

\begin{observation}\label{obs-vcorder}
Every $n$-vertex graph $G$ with a vertex cover $C$ of size $k$ admits a $k$-page book embedding with any vertex order $\prec$. Moreover, 
if $G$ and $C$ are given as input, such a book embedding can be computed in $\bigoh(n+k\cdot n)$ time.
\end{observation}

\begin{proof}
Let $C = \{c_1,\ldots,c_k\}$ be a vertex cover of size $k$ and let $\sigma$ be a page assignment on $k$ pages defined as follows. 
For each $i\in [k]$ all edges $u c_i$ with $u \in U\cup \{c_1,\dots,c_{i-1}\}$ are assigned to page $i$. 
Now, consider the edges assigned to any page $i \in [k]$. 
By construction, they are all incident to vertex $c_i$, and thus no two of them cross each other. 
Therefore, the pair $\langle \prec, \sigma \rangle$ is a $k$-page book embedding of $G$ and can be computed in $\bigoh(n+k\cdot n)$ time.\qed
\end{proof}

We note that the bound given in Observation~\ref{obs-vcorder} is tight, since it is known that complete bipartite graphs with bipartitions of size $k$ and $h>k(k-1)$ have book thickness $k$~\cite{DBLP:journals/jct/BernhartK79} and vertex cover number $k$.

We now proceed to a description of our algorithm. For ease of presentation, we will add to $G$ an additional vertex of degree $0$, add it to $U$, and place it at the end of $\prec$ (observe that this does not change the solution to the instance).

If $\tau \le k$ then we are done by \cref{obs-vcorder}. Otherwise, let $S$ be the set of all possible non-crossing page assignments of the edges whose both endpoints lie in $C$, and note that $|S|< \tau^{\tau^2}$ and $S$ can be constructed in time $\bigoh(\tau^{\tau^2})$ (recall that $k < \tau$ by assumption). As its first step, the algorithm branches over each choice of $s\in S$, where no pair of edges assigned to the same page crosses.

For each such non-crossing assignment $s$, the algorithm performs a
dynamic programming procedure that runs
on the vertices of the input graph in sequential (left-to-right) order. 
We will define a record set that the algorithm is going to compute for each individual vertex in left-to-right order. 
Let $c_1 \prec \ldots \prec c_\tau$ be the ordering of vertices of $C$,
and let $u_1 \prec \ldots \prec u_{n-\tau}$ be the ordering of vertices of $U$.

In order to formalize our records, we need the notion of visibility. Let $i\in [n-\tau]$ and let $E_i=\{u_j c\in E \mid j< i,\, c \in C\}$ be the set of all edges with one endpoint outside of $C$ that lies to the left of $u_i$. We call $\alpha \colon E_i\rightarrow [k]$ a \emph{valid partial page assignment} if $\alpha\cup s$ maps edges to pages in a non-crossing fashion.
Now, consider a valid partial page assignment $\alpha \colon E_i\rightarrow [k]$. We say that a vertex $c\in C$ is \emph{$(\alpha,s)$-visible} to $u_t$ (for $t\in [n-\tau]$) on page $p$ if it is possible to draw an edge from $u_t$ to $c$ on page $p$ without crossing any other edge mapped to page $p$ by $\alpha \cup s$. \cref{fig:visibility} shows the visibilities of a vertex in two pages.

Based on this notion of visibility, for an index $a\in [n-\tau]$ 
we can define a $k \times \tau$ \emph{visibility matrix}
$M_i(a,\alpha,s)$, where an entry $(p,b)$ of $M_i(a,\alpha,s)$ is $1$ if $c_b$ is $(\alpha, s)$-visible to $u_a$ on page $p$ and $0$ otherwise (see \cref{fig:visibility}). 
Intuitively, this visibility matrix captures information about the reachability via crossing-free edges (i.e., \emph{visibility}) to the vertices in $C$ from $u_a$ on individual pages given a particular assignment $\alpha$ of edges in $E_i$.
Note that for a given tuple $(i,a,\alpha,s)$, it is straightforward to compute %
$M_i(a,\alpha,s)$ in polynomial time.

\begin{figure}[tb]
	\centering
		\includegraphics[scale=1]{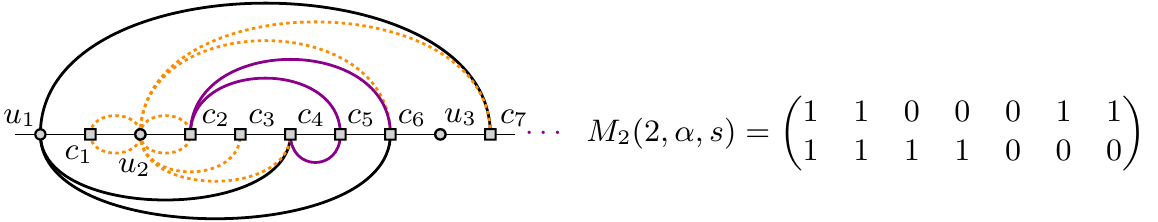}
	\caption{A partial 2-page book embedding of a graph $G$ with a vertex cover $C$ of size 7. The visibilities of vertices in $C$ (squares) from $u_2$ are marked by dashed edges (left). Corresponding visibility matrix $M_2(2,\alpha,s)$ %
	(right).\vspace{-0.5cm}}
	\label{fig:visibility}
\end{figure} 

Observe that while the number of possible choices of valid partial page assignments $\alpha \colon E_i\rightarrow [k]$ (for some $i\in [n-\tau]$) is not bounded by a function of $\tau$, for each $i,a\in [n-\tau]$ the number of possible visibility matrices is upper-bounded by $2^{\tau^2}$. 
On a high level, the core idea in the algorithm is to dynamically process the vertices in $U$ in a left-to-right fashion and compute, for each such vertex, a bounded-size ``snapshot'' of its visibility matrices---whereas for each such snapshot we will store only one (arbitrarily chosen) valid partial page assignment. We will later (in Lemma~\ref{lem:vcordercorrect}) show that all valid partial page assignments leading to the same visibility matrices are ``interchangeable''.

With this basic intuition, we can proceed to formally defining our records. Let $X=\{x\in [n-\tau] \mid \exists c\in C: u_x $ is the immediate successor of $c$ in $\prec\}$ be the  set of indices of vertices in $U$ which occur immediately after a cover vertex; we will denote the integers in $X$ as $x_1,\dots, x_z$ (in ascending order), and we note that $z\leq \tau$. For a vertex $u_i\in U$, we define our record set as follows: 
$\RRR_i(s)=\{ \big(M_i(i,\alpha,s),M_i(x_1,\alpha,s),M_i(x_2,\alpha,s),\dots,M_i(x_z,\alpha,s)\big) \mid
\exists \text{ valid partial page assignment } \alpha \colon E_i\rightarrow [k]\}$. Note that each entry in $\RRR_i(s)$ captures one possible set (a ``snapshot'') of at most $\tau+1$ visibility matrices: the visibility matrix for $u_i$ itself, and the visibility matrices for the $z$ non-cover vertices which follow immediately after the vertices in $C$. The intuition behind these latter visibility matrices is that they allow us to update our visibility matrix when our left-to-right dynamic programming algorithm reaches a vertex in $C$ (in particular, as we will see later, for $i\in X$ it is not possible to update the visibility matrix $M_i(i,\alpha,s)$ only based on $M_{i-1}(i-1,\alpha,s)$). Along with $\RRR_i(s)$, we also store a mapping $\Lambda_i^s$ from $\RRR_i(s)$ to valid partial page assignments of $E_i$ which maps $(M_0,\dots,M_{z})\in \RRR_i(s)$ to some $\alpha$ such that $(M_0,\dots,M_{z})=(M_i(i,\alpha,s),M_i(x_1,\alpha,s),M_i(x_2,\alpha,s),\dots,M_i(x_z,\alpha,s))$.

Let us make some observations about our records $\RRR_i(s)$. First, $|\RRR_i(s)|\leq 2^{\tau^3+\tau^2}$. 
Second, if $\RRR_{n-\tau}(s)\neq \emptyset$ for some $s$, since $u_{n-\tau}$ is a dummy vertex of degree 0, then there is a valid partial page assignment $\alpha \colon E_{n-\tau} \rightarrow [k]$ such that $s\cup \alpha$ is a non-crossing page assignment of \emph{all} edges in $G$. Hence we can output a $k$-page book embedding by invoking $\Lambda_{n-\tau}^s$ on any entry in $\RRR_{n-\tau}(s)$. Third:

\newcommand{\obstwo}{If for all $s\in S$ it holds that $\RRR_{n-\tau}(s)=\emptyset$, then $(G,\prec,k)$ is a NO-instance of \textsc{Fixed-Order Book Thickness}.}
\begin{observation}[*]\label{obstwo}
\obstwo
\end{observation}

The above implies that in order to solve our instance, it suffices to compute $\RRR_{n-\tau}(s)$ for each $s\in S$. As mentioned earlier, we do this dynamically, with the first step consisting of the computation of $\RRR_1(s)$. Since $E_1=\emptyset$, the visibility matrices $M_1(1,\emptyset,s),M_1(x_1,\emptyset,s),\dots,M_1(x_z,\emptyset,s)$ required to populate $\RRR_1(s)$ depend only on $s$ and are easy to compute in polynomial time.

Finally, we proceed to the dynamic step. Assume we have computed $\RRR_{i-1}(s)$. 
We branch over each possible page assignment $\beta$ of the (at most $\tau$) edges incident to $u_{i-1}$, and each tuple $\rho\in \RRR_{i-1}(s)$. 
For each such $\beta$ and $\gamma=\Lambda_{i-1}^s(\rho)$, we check whether $\beta\cup \gamma$ is a valid partial page assignment (i.e., whether $\beta\cup \gamma \cup s$ is non-crossing); if this is not the case, we discard this pair of $(\beta,\rho)$. Otherwise we compute the visibility matrices $M_i(i,\beta\cup \gamma,s), M_i(x_1,\beta\cup \gamma,s), \dots, M_i(x_z,\beta\cup \gamma,s)$, add the corresponding tuple into $\RRR_i(s)$, and set $\Lambda_i^s$ to map this tuple to $\beta\cup \gamma$. We remark that here the use of $\Lambda_{i-1}^s(\rho)$ allows us not to distinguish between $i\in X$ and $i\not \in X$---in both cases, the partial page assignment $\gamma$ will correctly capture the visibility matrix for $u_i$.

\begin{lemma}
\label{lem:vcordercorrect}
The above procedure correctly computes $\RRR_i(s)$ from $\RRR_{i-1}(s)$.
\end{lemma}

\begin{proof}
Consider an entry $(M_0,\dots,M_{z})$ computed by the above procedure from some $\beta\cup \gamma$. Since we explicitly checked that $\beta\cup \gamma$ is a valid partial page assignment, this implies that $(M_0,\dots,M_{z})\in \RRR_i(s)$, as desired.

For the opposite direction, consider a tuple $(M_0,\dots,M_{z})\in \RRR_i(s)$. By definition, there exists some valid partial page assignment $\alpha$ of $E_i$ such that $M_0=M_i(i,\alpha,s)$, $M_1=M_i(x_1,\alpha,s)$, \dots, $M_{z}=M_i(x_z,\alpha,s)$. Now let $\beta$ be the restriction of $\alpha$ to the edges incident to $u_{i-1}$, and let $\gamma'$ be the restriction of $\alpha$ to all other edges (i.e., all those not incident to $u_{i-1}$). Since $\gamma'\cup s$ is non-crossing and in particular $\gamma'$ is a valid partial page assignment for $E_{i-1}$, $\RRR_{i-1}(s)$ must contain an entry $\omega=\left(M_{i-1}(i-1,\gamma',s),\dots,(M_{i-1}(x_z,\gamma',s)\right)$---let $\gamma=\Lambda_{i-1}^s(\omega)$. 

To conclude the proof, it suffices to show that (1) $\beta\cup \gamma$ is a valid partial page assignment, and (2) $(M_i(i,\beta\cup\gamma',s),\dots,M_i(x_z,\beta\cup\gamma',s))$, which is the original tuple corresponding to the hypothetical $\alpha$, is \emph{equal} to $(M_i(i,\beta\cup\gamma,s),\dots,M_i(x_z,\beta\cup\gamma,s))$, which is the entry our algorithm computes from $\beta$ and $\gamma$. 
Point (1) follows from the fact that $M_{i-1}(i-1,\gamma',s)=M_{i-1}(i-1,\gamma,s)$ in conjunction with the fact that $u_{i-1}$ is adjacent only to vertices in $C$. Point (2) then follows by the same argument, but applied to each visibility matrix in the respective tuples: for each $x\in X$ we have $M_{i-1}(x,\gamma',s)=M_{i-1}(x,\gamma,s)$---meaning that the visibilities of $u_x$ were identical before considering the edges incident to $u_{i-1}$---and so assigning these edges to pages as prescribed by $\beta$ leads to an identical outcome in terms of visibility.
\qed
\end{proof}

This proves the correctness of our algorithm. The runtime is upper-bounded by the product of $|S|< \tau^{\tau^2}$ (the initial branching factor), $n$ (the number of times we compute a new record set $\RRR_{i}(s)$), and $2^{\tau^3+\tau^2}\cdot \tau^\tau$ (to consider all combinations of $\gamma$ and $\beta$ so to compute a new record set from the previous one).
A minimum-page book embedding can be computed by trying all possible choices for $k\in [\tau]$. We summarize Result 1 below.

\begin{theorem}
\label{thm:vcnfixed}
There is an algorithm which takes as input an $n$-vertex graph $G$ with a vertex order $\prec$, runs in time $2^{\bigoh(\tau^3)}\cdot n$ where $\tau$ is the vertex cover number of $G$, and computes a page assignment $\sigma$ such that $(\prec,\sigma)$ is a $(\fobt(G,\prec))$-page book embedding of $G$.
\end{theorem}

\subsection{Parameterization by the Pathwidth of the Vertex Ordering}\label{sec:pw_fixed} 

As our second result, we show that \textsc{Fixed-Order Book Thickness} is fixed-parameter tractable parameterized by the pathwidth of $(G,\prec)$. We note that while the pathwidth of $G$ is always upper-bounded by the vertex cover number, this does not hold when we consider a fixed ordering $\prec$, and hence this result is incomparable to \cref{thm:vcnfixed}. 
For instance, if $G$ is a path, it has arbitrarily large vertex cover number while $(G,\prec)$ may have a pathwidth of $1$, while on the other hand if $G$ is a star, it has a vertex cover number of $1$ while $(G,\prec)$ may have arbitrarily large pathwidth.
To begin, we can show that the pathwidth of $(G,\prec)$ provides an upper bound on the number of pages required for an embedding.

\newcommand{\lempathk}{Every $n$-vertex graph $G=(V,E)$ with a linear order $\prec$ of $V$ such that $(G,\prec)$ has pathwidth $k$ admits a $k$-page book embedding $\langle \prec, \sigma \rangle$, which can be computed in  $\bigoh(n+k\cdot n)$ time.}

\begin{lemma}[*]\label{lem-path-k}
\lempathk
\end{lemma}

We note that the bound given in \cref{lem-path-k} is also tight for the same reason as for \cref{obs-vcorder}: complete bipartite graphs with bipartitions of size $k$ and $h>k(k-1)$ have book thickness $k$~\cite{DBLP:journals/jct/BernhartK79}, but admit an ordering $\prec$ with pathwidth~$k$. 

We now proceed to a description of our algorithm. Our input consists of the graph $G$, the vertex ordering $\prec$, and an integer $k$ that upper-bounds the desired number of pages in a book embedding. 
Let $\kappa$ be our parameter, i.e., the pathwidth of $(G,\prec)$; observe that due to \cref{lem-path-k}, we may assume that $k \le \kappa$.
The algorithm performs a dynamic programming procedure 
on the vertices $v_1,v_2,\dots,v_n$ of the input graph $G$ in right-to-left order along $\prec$. 
For technical reasons, we initially add a vertex $v_0$ of degree $0$ to $G$ and place it to the left of $v_1$ in $\prec$; note that this does not increase the pathwidth of $G$. 

We now adapt the concept of \emph{visibility} introduced in \cref{sec:vc_fixed} for use in this algorithm. 
First, let us expand our notion of guard set (see \cref{preliminaries}) as follows: for a vertex $v_i$, let $P^*_{v_i}=\{g^i_1,\dots,g^i_m\}$ where for each $j\in [m-1]$, $g^i_j$ is the $j$-th guard of $v_i$ in reverse order of $\prec$ (i.e., $g^i_1$ is the guard that is nearest to $v_i$ in $\prec$), and $g^i_m=v_0$.
For a vertex $v_i$,
let $E_i=\{v_av_b \mid v_av_b\in E,\, b>i\}$ be the set of all edges with at least one endpoint to the right of $v_i$ and let $S_i=\{g^i_jv_b \mid g^i_j\in P^*_{v_i},\, g^i_jv_b\in E_i\}$
be the restriction of $E_i$ to edges between a vertex to the right of $v_i$ %
and a guard in $P^*_{v_i}$.
An assignment $\alpha \colon E_i\rightarrow [k]$
is called a \emph{valid partial page assignment} if $\alpha$ maps the edges in $E_i$ to pages in a non-crossing manner. 
Given a valid partial page assignment $\alpha \colon E_i\rightarrow [k]$ and a vertex $v_a$ with $a\leq i$, 
we say a vertex $v_x$ ($x<a$) is 
\emph{$\alpha$-visible} to $v_a$ on a page~$p$ if it is possible 
to draw the edge $ v_a v_x$ in page $p$ without crossing
any other edge mapped to $p$ by $\alpha$. 

Before we proceed to describing our algorithm, we will show that the visibilities of vertices w.r.t.\ valid partial page assignments exhibit a certain regularity property. Given $a\leq i\leq n$, $p\in [k]$, and a valid partial page assignment $\alpha$ of $E_i$, let the \emph{$(\alpha,i,p)$-important} edge of $v_a$ be the edge $v_cv_d\in S_i$ with the following properties: 
\begin{inparaenum}[(1)]
\item $\alpha(v_cv_d)=p$,
\item $c<a$, and
\item $|a-c|$ is minimum among all such edges in $S_i$.
\end{inparaenum}
If multiple edges with these properties exist, we choose the edge with minimum $|d-c|$. Intuitively, the $(\alpha,i,p)$-important edge of $v_a$ is simply the shortest edge of $S_i$ which encloses $v_a$ on page $p$; note that it may happen that $v_a$ has no $(\alpha,i,p)$-important edge. Observe that, if the edge exists, its left endpoint is $v_c\in P^*_{v_i}$, and we call $v_c$ the \emph{$(\alpha,i,p)$-important} guard of $v_a$.
The next observation easily follows from the definition of $(\alpha,i,p)$-important edge, see also \cref{fig:range-visbility}.

\begin{figure}[tb]
\centering
\includegraphics[scale=1]{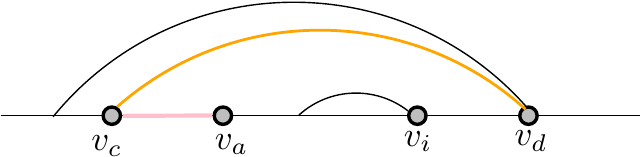}
\caption{An assignment of the edges of $S_i$ to a page~$p$, where the edge $v_cv_d$
is the $(\alpha,i,p)$-important edge of $v_a$. Any vertex $w$ with $v_c \prec w \prec v_a$ is visible to $v_a$, and any vertex $w' \prec v_c$ is not visible to $v_a$.}
\label{fig:range-visbility}
\end{figure}

\begin{observation}
\label{obs:pwvisibility}
If $v_a$ has no $(\alpha,i,p)$-important edge, then every vertex $v_x$ with $x<a$ is $\alpha$-visible to $v_a$. If the $(\alpha,i,p)$-important guard of $v_a$ is $v_c$, then $v_x$ $(x<a)$ is $\alpha$-visible to $v_a$ if and only if $x\geq c$.
\end{observation}

\cref{obs:pwvisibility} not only provides us with a way of handling vertex visibilities in the pathwidth setting, but also allows us to store all the information we require about vertex visibilities in a more concise way than via the matrices used in \cref{sec:vc_fixed}. For an index $i\in[n]$, a vertex $v_a$ where $a\leq i$ and a valid partial page assignment $\alpha$, we define the \emph{visibility vector} $U_i(v_a,\alpha)$ as follows: the $p$-th component of $U_i(v_a,\alpha)$ is the $(\alpha,i,p)$-important guard of $v_a$, and $\diamond$ if $v_a$ has no $(\alpha,i,p)$-important guard. Observe that since the number of pages is upper-bounded by $\kappa$ by assumption and the cardinality of $P^*_{v_i}$ is at most $\kappa+1$, there are at most $(\kappa+2)^\kappa$ possible distinct visibility vectors for any fixed $i$. 

Observe that thanks to \cref{obs:pwvisibility} the visibility vector $U_i(v_i,\alpha)$ provides us with complete information about the visibility of vertices $v_b$ ($b<i$) from $v_i$---notably, $v_b$ is not $\alpha$-visible to $v_i$ on page $p$ if and only if $v_b$ lies to the left of the $(\alpha, i, p)$-important guard $U_i(v_i,\alpha)[p]$ (and, in particular, if $U_i(v_i,\alpha)[p]=\diamond$ then every such $v_b$ is $\alpha$-visible to $v_i$ on page $p$). On a high level, the algorithm will traverse vertices in right-to-left order along $\prec$ and store the set of all possible visibility vectors at each vertex. To this end, it will use the following observation to update its visibility vectors.

\begin{observation}
\label{obs:vis-update}
Let $\alpha$ be a valid partial page assignment of $E_i$ and $p$ be a page. If $v_{i-1}\not\in P^*_{v_i}$, then a vertex $v_b$ $(b<i-1)$ is $\alpha$-visible to $v_{i-1}$ on page $p$ if and only if $v_b$ is $\alpha$-visible to $v_{i}$ on page $p$.
\end{observation}

\begin{proof}
By definition $v_{i-1}$ and $v_i$ are consecutive in $\prec$. 
Let $v_b$ (for $b<i-1$) be a vertex that is $\alpha$-visible to $v_{i-1}$ on page~$p$. If $v_b$ is not $\alpha$-visible to $v_i$ on $p$, then there must be a vertex $w$ between $v_{i-1}$ and $v_i$ that is 
incident to an edge in $E_i$ separating $v_{i-1}$ and $v_i$ on page $p$. But this contradicts that $v_{i-1}$ and $v_i$ are consecutive in $\prec$. The other direction follows by the same argument.
\qed\end{proof}

There is, however, a caveat: \cref{obs:vis-update} does not (and in fact cannot) allow us to compute the new visibility vector if $v_{i-1}\in P^*_{i}$. To circumvent this issue, our algorithm will not only store the visibility vector $U_i(v_i,\alpha)$ but also the visibility vectors for each guard of $v_i$. We now prove that we can compute the visibility vector for any vertex from the visibility vectors of the guards---this is important when updating our records, since we will need to obtain the visibility records for new guards that are introduced at some step of~the~algorithm. 

\begin{lemma}
\label{lem:pwvisibilitynew}
Let $v_a \prec v_i$, $\alpha$ be a valid partial page assignment of $E_i$, $p \in [k]$ be a page,
and assume $v_a\notin P^*_i$. 
Let $v_b\in P^*_i\cup \{v_i\}$ be such that $b>a$ and $|b-a|$ is minimized, i.e., $v_b$ is the first guard to the right of $v_a$. Then $U_i(v_a,\alpha)=U_i(v_b,\alpha)$.
\end{lemma}

\begin{proof}
Let $v_x$ for $x<a$ be any vertex that is $\alpha$-visible to $v_a$ in page $p$ and assume $v_x$ is not $\alpha$-visible to $v_b$. Then there must be an edge $wz \in E_i$ separating $v_a$ from $v_b$ in page $p$, i.e., $v_a \prec w \prec v_b$.  But in that case $w$ is a guard in $P^*_i$ closer to $v_a$ contradicting the choice of $v_b$. 
Conversely, let $v_x$ for $x<a$ be a vertex that is not $\alpha$-visible to $v_a$ in page $p$. Then there must be an edge $wz \in E_i$ separating $v_x$ from $v_a$ on page $p$. Then edge $wz$ also separates $v_x$ from $v_b$ and $v_x$ is not $\alpha$-visible to $v_b$.
Therefore, the visibility vectors $U_i(v_a,\alpha)$ and $U_i(v_b,\alpha)$ corresponding to the vertices $v_a$ and $v_b$, respectively, 
are equal.
\qed
\end{proof}

We can now formally define our record set as $Q_{i}=\{(U_i(v_i,\alpha), U_i(g_1^i,\alpha), \ldots,\\  U_i(g_{m-1}^i,\alpha)) \mid \exists \text{ valid partial page assignment } \alpha \colon E_i \rightarrow [k]\}$, where each individual element (record) in $Q_{i}$ can be seen as a queue starting with the visibility vector for $v_i$ and then storing the visibility vectors for individual guards (note that there is no reason to store an ``empty'' visibility vector for $g_m^i$). To facilitate the construction of a solution, we will also store a function $\Lambda_i$ from $Q_i$ to valid partial page assignments of $E_i$ which maps each tuple $\omega\in Q_i$ to some $\alpha$ such that $\omega=(U_i(v_i,\alpha), U_i(g_1^i,\alpha), \ldots,  U_i(g_{m-1}^i,\alpha))$.

Let us make some observations about our records $Q_{i}$. First of all, since there are at most $(\kappa+2)^{\kappa}$ many visibility vectors, $|Q_i|\leq (\kappa+2)^{\kappa^2}$. Second, if $|Q_0|>0$ then, since $E_0=E$, the mapping $\Lambda_0(\omega)$ will produce a valid page assignment of $E$ for any $\omega\in Q_0$. On the other hand, if $G$ admits a $k$-page book embedding $\alpha$ with order $\prec$, then $\alpha$ witnesses the fact that $Q_0$ cannot be empty. Hence, the algorithm can return one, once it correctly computes $Q_0$ and $\Lambda_0$.

The computation is carried out dynamically and starts by setting $Q_n=\{\omega\}$, where $\omega=(\diamond)$, and $\Lambda_n(\omega)=\emptyset$. 
For the inductive step, assume that we have correctly computed $Q_i$ and $\Lambda_i$, and the aim is to compute $Q_{i-1}$ and $\Lambda_{i-1}$. For each $\omega=(\omega_1,\dots,\omega_m)\in Q_i$, we compute an intermediate record $\omega'$ which represents the visibility vector of $v_{i-1}$ w.r.t.\ $\alpha=\Lambda_i(\omega)$ as follows:

\begin{itemize}
\item if $v_{i-1}\in P^*_{i}$, then $\omega'=(\omega_2,\dots,\omega_m)$, and
\item if $v_{i-1}\not \in P^*_{i}$, then $\omega'=(\omega_1,\dots,\omega_m)$ (Recall \cref{obs:vis-update}).
\end{itemize}

We now need to update our intermediate record $\omega'$ to take into account the new guards. In particular, we expand $\omega'$ by adding, for each new guard $g^{i-1}_j\in P^*_{i-1}\setminus P^*_{i}$, an intermediate visibility vector $U_{i-1}(g^{i-1}_j,\alpha)$ at the appropriate position in $\omega'$ (i.e., mirroring the ordering of guards in $P^*_{i-1}$). Recalling \cref{lem:pwvisibilitynew}, we compute this new intermediate visibility vector $U_{i-1}(g^{i-1}_j,\alpha)$ by copying the visibility vector that immediately succeeds it in $\omega'$. 

Next, let $F_{i-1}=E_{i-1}\setminus E_i$ be the at most $\kappa$ new edges that we need to account for, and let us branch over all assignments $\beta \colon F_{i-1}\rightarrow [k]$. For each such $\beta$, we check whether $\alpha\cup\beta$ is a valid partial page assignment of $E_{i-1}$, i.e., whether the new edges in $F_{i-1}$ do not cross with each other or other edges in $E_i$ when following the chosen assignment $\beta$ and the assignment $\alpha$ obtained from $\Lambda_i$. As expected, we discard any $\beta$ such that $\alpha\cup\beta$ is not valid.

Our final task is now to update the intermediate visibility vectors $U_{i-1}(*,\alpha)$ (with $*$ being a placeholder) to $U_{i-1}(*,\alpha\cup \beta)$. This can be done in a straightforward way by, e.g., looping over each edge $e\in F_{i-1}$, obtaining the page $p=\beta(e)$ that $e$ is mapped to, reading $U_{i-1}(*,\alpha)[p]$ and replacing that value by the guard $g$ incident to $e$ if $g$ occurs to the right of $U_{i-1}(*,\alpha)[p]$ and to the left of $*$. Finally, we enter the resulting record $\omega'$ into $Q_{i-1}$.

\begin{lemma}
\label{lem:vcordercorrect-2}
The above procedure correctly computes $Q_{i-1}$ from $Q_i$.
\end{lemma}

\begin{proof}
Consider an entry $\omega'$ computed by the above procedure from some $\alpha\cup \beta$ and $\omega$. Since we explicitly checked that $\alpha\cup \beta$ is a valid partial page assignment for $E_{i-1}$, there must exist a record $(U_{i-1}(v_{i-1},\alpha\cup \beta), U_{i-1}(g_1^{i-1},\alpha\cup\beta), \dots, U_{i-1}(g_{m-1}^{i-1}))\in Q_{i-1}$, and by recalling \cref{obs:pwvisibility}, \cref{lem:pwvisibilitynew} and \cref{obs:vis-update} it can be straightforwardly verified that this record is equal to~$\omega'$.

For the opposite direction, consider a tuple $\omega_0\in Q_{i-1}$ that arises from the valid partial page assignment $\gamma$ of $E_{i-1}$, and let $\beta$, $\alpha$ be the restrictions of $\gamma$ to $F_{i-1}$ and $E_i$, respectively. Since $\alpha$ is a valid partial page assignment of $E_i$, there must exist a tuple $\omega\in Q_i$ that arises from $\alpha$. Let $\alpha'=\Lambda_i(\omega)$. To conclude the proof, it suffices to note that during the branching stage the algorithm will compute a record from a combination of $\alpha'$ (due to $\omega$ being in $Q_i$) and $\beta$, and the record computed in this way will be precisely $\omega_0$.
\qed\end{proof}

This proves the correctness of the algorithm. 
The runtime is upper bounded by $\bigoh(n\cdot (\kappa+2)^{\kappa^2}\cdot \kappa^\kappa)$ (the product of the number of times we compute a new record, the number of records and the branching factor for $\beta$). 
A minimum-page book embedding can be obtained by trying all possible choices for $k\in [\kappa]$).
We summarize Result 2 below.

\begin{theorem}
There is an algorithm which takes as input an $n$-vertex graph $G=(V,E)$ with a vertex ordering $\prec$ and computes a page assignment $\sigma$ of $E$ such that $(\prec,\sigma)$ is a $(\fobt(G,\prec))$-page book embedding of $G$. The algorithm runs in $n\cdot \kappa^{\bigoh(\kappa^2)}$ time where $\kappa$ is the pathwidth of $(G,\prec)$.
\end{theorem}

\section{Algorithms for \textsc{Book Thickness}}\label{without-ordering}

We now turn our attention to the general definition of book thickness (without a fixed vertex order). 
We show that, given a graph $G$, in polynomial time we can construct an equivalent instance $G^*$ whose size is upper-bounded by a function of $\tau(G)$. Such an algorithm is called a \emph{kernelization} and directly implies the fixed-parameter tractability of the problem with this parameterization~\cite{DBLP:books/sp/CyganFKLMPPS15,DBLP:series/txcs/DowneyF13}.

\begin{theorem}\label{th:vc-main}
There is an algorithm which takes as input an $n$-vertex graph $G=(V,E)$ and a positive integer $k$, 
runs in time $\bigoh(\tau^{{\tau^{\bigoh(\tau)}}} + 2^\tau \cdot n)$ where $\tau=\tau(G)$ is the vertex cover number of $G$, and
decides whether $\bt(G) \le k$. If the answer is positive, it can also output a $k$-page book embedding of $G$.
\end{theorem}
\begin{proof}
If $k > \tau$, by \cref{obs-vcorder} we can immediately conclude that $G$ admits a $k$-page book embedding. Hence we shall assume that $k \le \tau$. We will also compute a vertex cover $C$ of size $\tau$ in time $\bigoh(2^\tau\cdot n)$ using well-known results~\cite{DBLP:journals/tcs/ChenKX10}.

For any subset $U \subseteq C$ we say that a vertex of  $V \setminus C$ is of \emph{type $U$} if its set of neighbors is equal to $U$. This defines an equivalence relation on $V \setminus C$ and partitions $V \setminus C$ into at most $\sum_{i=0}^{\tau} {\tau \choose{i}}=2^\tau$ distinct types. 
In what follows, we denote by $V_U$ the set of vertices of type $U$. 
We claim the following.

\begin{claim}\label{cl:kernel}
Let $v \in V_U$ such that $|V_U| \ge 2 \cdot k^\tau + 2$.
Then $G$ admits a $k$-page book embedding if and only if $G'=G \setminus \{v\}$ does. Moreover, a $k$-page book embedding of $G'$ can be extended to such an embedding for $G$ in linear time.
\end{claim}

\begin{proof}[of the Claim]
One direction is trivial, since removing a vertex from a book embedding preserves the property of being a book embedding of the resulting graph. 
So let $\langle \prec, \sigma \rangle$ be a $k$-page book embedding of $G'$. 
We prove that a $k$-page book embedding of $G$ can be easily constructed by inserting $v$ right next to a suitable vertex $u$ in $V_U$ and by assigning the edges of $v$ to the same pages as the corresponding edges of $u$. 
We say that two vertices $u_1, u_2 \in V_U$ are \emph{page equivalent}, if for each vertex $w \in U$, the edges $u_1w$ and $u_2w$ are both assigned to the same page according to $\sigma$. 
Each vertex in $V_U$ has degree exactly $|U|$, hence this relation partitions the vertices of $V_U$ into at most 
$k^{|U|} \le k^\tau$ sets. 
Since $|V_U| \setminus \{v\} \ge 2 \cdot k^\tau + 1$, at least three vertices of this set, which we denote by $u_1$, $u_2$, and $u_3$,  are page equivalent. 
Consider now the graph induced by the edges of these three vertices that are assigned to a particular page. 
By the above argument, such a graph is a $K_{h,3}$, for some $h>0$. However, since already $K_{2,3}$ does not admit a $1$-page book embedding, we have $h \le 1$, that is, each $u_i$ has at most one edge on each page. 
Then we can extend $\prec$ by introducing $v$ right next to $u_1$ and assign each edge $vw$ to the same page as $u_1w$. 
Since each such edge $vw$ runs arbitrarily close to the corresponding crossing-free edge $u_1 w$, this results in a $k$-page book embedding of $G$ and concludes the proof of the claim.~\qed
\end{proof}

\medskip 
We now construct a kernel $G^*$ from $G$ of size $\bigoh(k^{\tau})$ as follows. 
We first classify each vertex of $G$ based on its type. 
We then remove an arbitrary subset of vertices from each set $V_U$ with $|V_U| > 2 \cdot k^\tau + 1$ until $|V_U| = 2 \cdot k^\tau + 1$. 
Thus, constructing $G^*$ can be done in $\bigoh(2^\tau + \tau \cdot n)$ time, where $2^\tau$ is the number of types and $\tau \cdot n$ is the maximum number of edges of $G$.
From our claim above we can conclude that $G^*$ admits a $k$-page book embedding if and only if $G$ does. 
Determining the book thickness of $G^*$ can be done by guessing all possible linear orders and and page assignments in $O(k^\tau ! \cdot k^{k^\tau})=O(\tau^{\tau^{O(\tau)}})$ time. 
A $k$-page book embedding of $G^*$ (if any) can be extended to one of $G$ by iteratively applying the constructive procedure from the proof of the above claim, in $O(\tau \cdot n)$ time.~\qed
\end{proof}

The next corollary easily follows from \cref{th:vc-main}, by applying a binary search on the number of pages $k \le \tau$ and by observing that a vertex cover of minimum size $\tau$ can be computed in $2^{\bigoh(\tau)}+\tau\cdot n$ time~\cite{DBLP:journals/tcs/ChenKX10}.

\begin{corollary}\label{co:vc-min}
Let $G$ be a graph with $n$ vertices and vertex cover number $\tau$. A book embedding of $G$ with minimum number of pages can be computed in $\bigoh(\tau^{{\tau^{\bigoh(\tau)}}}+\tau \log \tau \cdot n)$ time.
\end{corollary}

\section{Conclusions and Open Problems}\label{conclusions}
We investigated the parameterized complexity of \textsc{Book Thickness} and \textsc{Fixed-Order Book Thickness}. We proved that both problems can be parameterized by the vertex cover number of the graph, and that the second problem can be parameterized by the pathwidth of the fixed linear order. The algorithm for \textsc{Book Thickness} is the first fixed-parameter algorithm that works for general values of $k$, while, to the best of our knowledge, no such algorithms were known for \textsc{Fixed-Order Book Thickness}.  

We believe that our techniques can be extended to the setting in which we allow edges on the same page to cross, with a given budget of at most $c$ crossings over all pages. This problem has been studied by Bannister and Eppstein~\cite{DBLP:journals/jgaa/BannisterE18} with the number of pages $k$ restricted to be either $1$ or $2$. It would also be interesting to investigate the setting where an upper bound on the maximum number of crossings \emph{per edge} is given as part of the input, which is studied in~\cite{DBLP:journals/ejc/BinucciGHL18}.

The main question that remains open is whether \textsc{Book Thickness} (and \textsc{Fixed-Order Book Thickness}) can be solved in polynomial time (and even fixed-parameter time) for graphs of bounded treewidth, which was asked by Dujmovi{\'c} and Wood~\cite{DBLP:journals/dmtcs/WoodD11}. As an intermediate step towards solving this problem, we ask whether the two problems can be solved efficiently when parameterized by the treedepth~\cite{DBLP:books/daglib/0030491} of the graph. Treedepth restricts the graph structure in a stronger way than treewidth, and has been used to obtain algorithms for several problems which have proven resistant to parameterization by treewidth~\cite{DBLP:journals/ai/GanianO18,DBLP:journals/siamdm/GutinJW16}.

\bibliographystyle{splncs04}
\bibliography{GD-ref}

\clearpage

\appendix

\section{Missing Proofs of \cref{fixed-ordering}}

\setcounter{observation}{1}
\begin{observation}
\obstwo
\end{observation}
\begin{proof}
Assume for a contradiction that $G$ admits a $k$-page book embedding with order $\prec$. Let $s$ be the restriction of that book embedding to edges whose both endpoints lie in $C$, and let $\alpha$ be the restriction of that book embedding to all other edges. Then $\alpha$ is a valid page assignment, and hence by definition $(M_{n-\tau}(n-\tau,\alpha,s),M_{n-\tau}(x_1,\alpha,s),M_{n-\tau}(x_2,\alpha,s),\dots,M_{n-\tau}(x_z,\alpha,s) ) \in \RRR_{n-\tau}(s)$. In particular, $\RRR_{n-\tau}(s)\neq \emptyset$.
\qed
\end{proof}

\setcounter{lemma}{1}
\begin{lemma}
\lempathk
\end{lemma}

\begin{proof}
Let $\sigma$ be the page assignment to $[k]$ defined as follows.
We parse the vertices of $G$ following $\prec$ from right to left. 
Consider the rightmost vertex $v_n$ of $\prec$, and let $U_n$ be an arbitrary 
injective assignment from $P_n$ (the guard set of $v_n$) to $[k]$. 
For each edge $e$ incident to $v_n$ there exists 
some $p\in P_n$, and we assign $e$ to page $U_n(p)$. 
Observe that, this results in a non-crossing page assignment of the edges incident to $v_n$. 

Next, we proceed by induction. Assume, we have obtained a non-crossing page assignment for all edges incident to the last $i$ vertices, i.e., for all edges incident to $\{v_j | j\geq n-i\}$, and that furthermore we have a mapping $U_{n-i}$ which maps the guards $P_{n-i}$ for $v_{n-i}$ to $[k]$ and a non-crossing partial page assignment which maps all edges $pv_j$ where $p\in P_{n-i}$ and $j\geq n-i$ to $U_{n-i}(p)$. In particular, all edges with an endpoint to the left of $v_{n-i}$ end in the guards for $v_{n-i}$ and are assigned to distinct pages if and only if they are incident to distinct guards.

We extend this page assignment to all edges incident to the last $i+1$ vertices as follows. First, we extend $U_{n-i}$ to an arbitrary injective mapping $U_{n-i-1}$, which is always possible since the number of guards for $v_{i+1}$ is at most $k$. Second, we assign each left edge $e=v_{n-i-1}p$ of $v_{n-i-1}$ to $U_{n-i-1}(p)$.

To conclude the proof, observe that the page assignment obtained in this way is non-crossing. Indeed, the only edges added to the current page assignment are left edges of $v_{n-i-1}$, and each such edge $e=v_{n-i-1}p$ is assigned to the page $U_{n-i-1}(p)$---notably, they maintain the property of being assigned to distinct pages if and only if they are incident to distinct guards. Also, it can be computed in $\bigoh(n+k\cdot n)$ time. 
\qed\end{proof}

\end{document}